\documentclass[conference]{IEEEtran}

\makeatletter
\def\ps@headings{%
\def\@oddhead{\mbox{}\scriptsize\rightmark \hfil \thepage}%
\def\@evenhead{\scriptsize\thepage \hfil \leftmark\mbox{}}%
\def\@oddfoot{}%
\def\@evenfoot{}}
\makeatother

\newcommand{\nop}[1]{}
\usepackage{stmaryrd}
\usepackage{bbding}
\usepackage{algorithmic}
\usepackage{algorithm}
\usepackage{graphicx}
\usepackage{epsfig}
\usepackage{subfigure}
\usepackage{multirow}
\usepackage{balance}

\usepackage[centertags]{amsmath}

\newtheorem{definition}{Definition}
\newtheorem{remark}{Remark}

\newtheorem{lemma}{Lemma}
\newtheorem{theorem}{Theorem}

\newtheorem{proof}{Proof}

\def\linespaces{0.95}
\def\baselinestretch{\linespaces}

\input epsf

\begin{document}

\title{A Stochastic Calculus for Network Systems with Renewable Energy Sources}

\author{\IEEEauthorblockN{Kui Wu}
\IEEEauthorblockA{Dept. of Computer Science\\
University of Victoria\\
Victoria, British Columbia, Canada}
\and
\IEEEauthorblockN{Yuming Jiang}
\IEEEauthorblockA{Q2S Center of Excellence \\
Norwegian University of Science and Technology\\
Trondheim, Norway
}
\and
\IEEEauthorblockN{Dimitri Marinakis}
\IEEEauthorblockA{Dept. of Computer Science\\
University of Victoria\\
Victoria, British Columbia, Canada}
}


\maketitle

\begin{abstract}
We consider the performance modeling and evaluation of network systems powered with renewable energy sources such as solar and wind energy. Such energy sources largely depend on environmental conditions, which are hard to predict accurately. As such, it may only make sense to require the network systems to support a soft quality of service (QoS) guarantee, i.e., to guarantee a service requirement with a certain high probability. In this paper,  we intend to build a solid mathematical foundation to help better understand the stochastic energy constraint and the inherent correlation between QoS and the uncertain energy supply. We utilize a calculus approach to model the cumulative amount of charged energy and the cumulative amount of consumed energy. We derive upper and lower bounds on the remaining energy level based on a stochastic energy charging rate and a stochastic energy discharging rate. By building the bridge between energy consumption and task execution (i.e., service), we study the QoS guarantee under the constraint of uncertain energy sources. We further show how performance bounds can be improved if some strong assumptions can be made.  
\end{abstract}



\begin{IEEEkeywords} Stochastic Network Calculus, Performance Evaluation, Renewable Energy, Energy Scheduling
\end{IEEEkeywords}

\section{Introduction}\label{sec:introduction}

In the last decades, there have been increasing demands on computing, communication and storage systems. The strong demands drive modern IT infrastructures to extend and scale at an unprecedented speed, raising serious ``green'' related concerns about high energy consumption and greenhouse emission. Addressing this problem, the use of renewable energy sources, such as solar and wind energy, plays an important role in the support of sustainable computing~\cite{green}. Various energy harvesting devices have been developed and broadly used in real-world applications. For example, to support perpetual environmental monitoring, solar energy has been used to power tiny sensor nodes deployed in the wilderness~\cite{Chou09}.  


While the benefit of using renewable energy is clear, its application can pose great challenges in terms of assuring a satisfactory quality of services (QoS) for computing systems. First, renewable energy sources are generally unreliable and hard to predict. The uncertainty in the energy supply makes it extremely hard to support strict QoS requirements. As a result, it only makes sense to demand the system to support a soft QoS guarantee, i.e., to guarantee the QoS requirement (e.g., delay, throughput) with a certain probability. Second, although modern computing and communication devices could be equipped with rate-adaptive capabilities~\cite{Zafer09}, allowing the devices to ``smartly'' schedule the execution of various tasks, it is nevertheless non-trivial to design good scheduling algorithms that can effectively utilize a limited and variable energy resource. Many aspects come into play, including for example, the predicted energy charging rate, the tolerable range of a QoS guarantee, the rate-adaptive features of the device, and so on. Third, QoS support and task scheduling become even harder when the network includes multiple nodes powered by renewable energy, because the task scheduling and energy management of one node may have direct impact on the other nodes. Such dependency makes QoS support extremely hard.

Given these challenges, we perceive the strong need for a generic analytical framework for performance modeling and evaluation of a network system using renewable energy sources. Such a theoretical model should (1) capture the stochastic features in energy replenishment and energy consumption, (2) provide good guidelines on the schedulability of given tasks under uncertain energy constraints, (3) be applicable to a large group of systems where the power-rate function~\cite{kansal2007power,Zafer09} may be of different forms, and (4) be able to analyze a single node as well as a network system.           

Various methods have been developed for task scheduling and performance analysis of energy constrained system, including for example Markov chain based methods~\cite{niyato2007sleep,susu2008stochastic}, calculus approaches~\cite{kansal2007power,Zafer09}, prediction-based approaches~\cite{Gorlatova11,vigorito2007adaptive}, and so on.  Nevertheless, no analytical model so far is sufficient to meet all the above requirements. We are thus motivated to develop a more generic analytical framework to fill the vacancy. In this paper, we make the following contributions: 

\begin{itemize}
\item We build a theoretical model for performance evaluation of network systems with renewable energy sources. Our model is based on recent progress in stochastic network calculus, but it much extends the concepts of the traditional stochastic network calculus by introducing energy charging/discharging models. Such extension is non-trivial since the new concepts require special treatment in the derivation of performance bounds. 
\item We derive the stochastic upper and lower bounds on a node's residual energy level. These bounds provide fundamental guidance in the energy management of a system with renewable energy supply. 
\item By using a generic power-rate function to bridge the task execution and its energy consumption, we derive the stochastic performance bounds on delay and system backlog. We study both the single-node case and the network case.  
\item We extend the model to analyze systems that utilize multiple energy sources. 
\item We point out a method to improve the performance bounds if some strong assumptions, such as independence between multiple energy sources, can be made. 
\end{itemize}

The rest of the paper is organized as follows. We introduce the basic notations in Section~\ref{sec:background}. In Section~\ref{sec:energyModel}, we present new energy models to capture the stochastic features in the energy charging and discharging processes. Based on the new energy models, we derive the stochastic bounds on a node's remaining energy. We analyze the performance bounds of a single node with respect to delay and system backlog in Section~\ref{sec:singleNode}. The performance of a network system is provided in Section~\ref{sec:network}. Methods of handling multiple energy sources and improving performance bounds are introduced in Sections~\ref{sec:multipleSources} and~\ref{sec:improvement}, respectively. Related work is discussed in Section~\ref{sec:relatedwork}. The paper is concluded in Section~\ref{sec:conclusion}. 

\section{Basic Notations} \label{sec:background}

We first introduce the basic notations following the convention of stochastic network calculus~\cite{Jia,Jiangbook,Li}. We denote by $\mathcal{F}$ the set of non-negative, wide-sense increasing functions, i.e.,  $$\mathcal{F}= \{f(\cdot): \forall 0\le x\le y, 0\le f(x) \le f(y)\},$$  and by $\bar{\mathcal{F}}$ the set of non-negative, wide-sense decreasing functions, i.e.,  $$\bar{\mathcal{F}}= \{f(\cdot): \forall 0\le x\le y, 0\le f(y) \le f(x)\}.$$ For any random variable $X$, its distribution function, denoted by $$F_X (x) \equiv Prob\{X\le x\},$$ belongs to $\mathcal{F}$, and its complementary distribution function, denoted by $$\bar{F}_X(x)\equiv Prob\{X>x\},$$ belongs to $\bar{\mathcal{F}}$. 

For any function $f(t)$, we use $f'(t)$ to denote its derivative, if it exists.

The following operations will be used in this paper:

\begin{itemize}
\item 
The $(\min,+)$ \textit{convolution} of functions $f$ and $g$ under the $(\min,+)$ algebra~\cite{Chang00, Cruz91a,Le} is defined as: 
\begin{equation}
(f\otimes g) (t) \equiv \inf_{0\le s \le t}\{f(s) + g(t-s)\}.
\end{equation}
\item
The $(\min,+)$ \textit{deconvolution} of functions $f$ and $g$ is defined as: 
\begin{equation}
(f\oslash g) (t) \equiv \sup_{s\ge0}\{f(t+s) - g(s)\}.
\end{equation}
\item The $(\max,+)$ \textit{convolution} of functions $f$ and $g$ is defined as:
\begin{equation}
(f \bar{\otimes} g) (t) \equiv \sup_{0 \le s \le t}\{f(s) + g(t-s)\}.
\end{equation}
\item
 The \textit{pointwise infimum} or
\textit{pointwise minimum} of functions $f$ and $g$
is defined as
\begin{equation}
(f\wedge g) (t) \equiv \min[f(t),g(t)].
\end{equation}
\end{itemize}
In addition, we adopt:
\begin{itemize}
\item $[x]^+ \equiv max\{x,0\}$,
\item $[x]_1 \equiv min\{x, 1\}$.
\end{itemize}

Throughout this paper, we assume that the \textit{energy charging curve} and the \textit{energy discharging curve}, both of which will be defined later, are non-negative and wide-sense increasing functions. In this paper, $C(t)$ and $C^*(t)$ are used to denote the \textit{cumulative} energy amount that has been charged and depleted in the time interval $(0,t]$, respectively. For any $0\le s \le t$, let $C(s,t) \equiv C(t)-C(s)$ and $C^*(s,t) \equiv C^*(t)-C^*(s).$ Similarly, $A(t)$ and $A^*(t)$ are used to denote the \textit{cumulative} amount of data traffic that has arrived and departed in time interval $(0,t]$, respectively, and $S(t)$ is used to denote the cumulative amount of service provided by the system in time interval $(0,t]$. For any $0\le s \le t$, let $A(s,t) \equiv A(t)-A(s), A^*(s,t) \equiv A^*(t)-A^*(s),$ and $S(s,t) \equiv S(t)-S(s).$ By default, $A(0)=A^*(0)=S(0)=0$, and $C(0)=C^*(0)=0$. 

\section{Stochastic Energy Models and Residual Energy} \label{sec:energyModel}
\subsection{Stochastic Energy Charging and Discharging Models}
\begin{definition}\label{def-arrival}\textbf{The stochastic energy charging model: }
The cumulative energy  amount $C(t)$ is said to follow \textit{\underline{s}tochastic \underline{e}nergy \underline{c}harging (s.e.c.)} curves $\alpha_1 \in \mathcal{F}$ and $\alpha_2 \in \mathcal{F}$ with bounding functions $f_1\in \bar{\mathcal{F}}$ and $f_2\in \bar{\mathcal{F}}$, respectively, denoted by $$C\sim_{sec}<f_1,\alpha_1, f_2, \alpha_2>,$$ if for all $t \ge s \ge 0$ and all $x\ge 0$, 
\begin{equation}
Prob\{\inf_{0\le s \le t}[C(s,t)- \alpha_1(t-s) > x]\} \ge f_1(x), \label{eq:sec1}
\end{equation}
and
\begin{equation}
Prob\{\sup_{0\le s \le t}[C(s,t)- \alpha_2(t-s) > x]\} \le f_2(x). \label{eq:sec2}
\end{equation}
We call the curves $\alpha_1$ and $\alpha_2$ the lower curve and the upper curve, respectively, and the functions $f_1$ and $f_2$ the lower bounding function and the upper bounding function, respectively.   
\end{definition}

\begin{remark}The practical meaning of (\ref{eq:sec1}) is to lower bound the cumulative amount of charged energy, i.e., the energy harvester should provide a minimal level of energy with a high probability. The meaning of (\ref{eq:sec2}) is that the cumulative amount of charged energy may be upper bounded. For instance, in the case of solar-powered sensor nodes, the total charged energy should not beyond the best case scenario, e.g., sunny all the time. 
\end{remark}

\begin{remark} The $(\rho, \sigma_1, \sigma_2$)-source in the \textit{harvesting theory} ~\cite{kansal2007power} is a special case of our stochastic energy charging model, by defining $\alpha_1(t) = \rho t - \sigma_1$, $\alpha_2(t) = \rho t + \sigma_2$, and setting the bounding functions to $0$. 
\end{remark}

\begin{definition}\label{def-service}\textbf{The stochastic energy discharging model:} Under the constraint of charged energy amount $C(t)$, the cumulative discharged energy amount $C^*(t)$ is said to have \textit{\underline{s}tochastic \underline{e}nergy \underline{d}ischarging (s.e.d.)} curves $\beta_1 \in \mathcal{F}$ and $\beta_2 \in \mathcal{F}$ with bounding functions $g_1 \in \bar{\mathcal{F}}$ and $g_2 \in \bar{\mathcal{F}}$, respectively, denoted by $$C^*\sim_{sed}<g_1,\beta_1, g_2, \beta_2>,$$ if for all $t\ge 0$ and all $x\ge 0$,  
\begin{equation}
Prob\{C \otimes \beta_1(t) - C^*(t) > x\} \ge g_1(x), \label{eq:sed1}
\end{equation}
and 
\begin{equation}
Prob\{C \otimes \beta_2(t) - C^*(t) > x\} \le g_2(x). \label{eq:sed2}
\end{equation}
We call curves $\beta_1$ and $\beta_2$ the upper curve and the lower curve, respectively, and the functions $g_1, g_2$ the upper bounding function and the lower bounding function, respectively.
\end{definition} 

\begin{remark} The practical meaning of the above definition is as follows. If we use $\beta_1(t)$ to denote the cumulative \textit{virtual} energy discharging amount (i.e., the energy budget that a service should follow), and if we want to use $\beta_1(t)$ to upper bound the actual discharged energy $C^*(t)$, then under the constraint of cumulative charged energy amount $C(t)$, at any time instance $t$, we have \begin{equation}
C^*(t) \le \beta_1(s)+ C(t-s) \label{eq:relation}.
\end{equation}
Note that we should not replace $\beta_1(s)$ with $C^*(s)$ since otherwise we obtain $C^*(t-s)\le C (t-s)$ for any time interval $(s,t]$, which is a constraint too restrictive. Because the inequality~(\ref{eq:relation}) holds for any $s \le t$, we have $C^*(t) \le \inf_{0\leq s \leq t}\{\beta_1(s)+ C(t-s)\}$, which is $C^*(t) \le C\otimes \beta_1(t)$ by the definition of (min,+) convolution. Inequality (\ref{eq:sed1}) in Definition~\ref{def-service} represents the stochastic version of the above relationship. In practice, we may also want to lower bound the discharged energy (with a curve $\beta_2$) to effectively utilize the harvested energy, since overcharged energy beyond the energy storage capacity will be wasted. In this case, the stochastic lower bound is represented by the inequality (\ref{eq:sed2}). 
\end{remark}

Since $C(t)$ and $C^*(t)$ denote the cumulative amount of charged energy and the cumulative amount of discharged energy, respectively, the remaining energy amount, denoted by $E(t)$, in the system at time $t$ can be calculated as:
\begin{equation}
E(t) = C(t) - C^*(t). \label{eq:Et}
\end{equation}

\subsection{Stochastic Analysis on Remaining Energy} \label{sec:SARE}

For any energy-aware scheduling, we first need to understand the \textit{feasible energy depletion region}. Particularly, we need to answer the following question: \textit{what are the lower and upper bounds of $Prob\{E(t) > x\}$?} 

The practical meaning of the lower bound comes from reliability consideration: the remaining energy at any time instance should be with a high probability larger than a minimal threshold value (called \textit{safety threshold}) to guarantee the reliable operation of the system. The upper bound is related to the effective use of the energy source: since the capacity of energy storage is limited (in this case, $x$ could be considered as the upper limit of the energy capacity), tasks should be scheduled to consume energy effectively so that the probability that the system is overcharged is small.

To answer the above question, we start with the following lemma. 

\begin{lemma}\label{lemma1} 
Assume that the \textit{complementary cumulative distribution functions (CCDF)} of random variables $X, Y$ are $\bar{F}_{X} \in \bar{\mathcal{F}}$  and $\bar{F}_{Y} \in \bar{\mathcal{F}}$, respectively. Assume that $f_1(x) \leq \bar{F}_{X}(x) \leq f_2(x)$ and $g_1(x) \leq \bar{F}_{Y}(x) \leq g_2(x)$. Denote $Z=X + Y$. Then no matter whether $X$ and $Y$ are independent or not, there holds for $\forall x \ge 0$, 
\begin{equation} 
\bar{F}_Z(x) \leq f_2 \otimes g_2(x), \label{eq:lemma1-1}
\end{equation} 
and 
\begin{equation} 
\bar{F}_Z(x) \geq f_1 \bar{\otimes} g_1(x) -1. \label{eq:lemma1-2}
\end{equation} 
\end{lemma}

The proofs of all lemmas are referred to Appendix of this paper. 

\begin{theorem}\label{theorem1} 
Assume that the system has an energy charging source with cumulative charged energy amount $C(t) \sim_{sec}<f_1,\alpha_1, f_2,\alpha_2>$ and it provides service with cumulative depleted energy amount $C^*(t) \sim_{sed}<g_1,\beta_1, g_2, \beta_2>$. The remaining energy amount of the system at any time instant $t$, $E(t)$, is lower bounded by 
\begin{equation}
Prob\{E(t) > x\} \geq f_1 \bar{\otimes} g_1 (x-\alpha_1 \oslash \beta_1 (0)) -1, \label{theorem1-1}
\end{equation}
and is upper bounded by 
\begin{equation}
Prob\{E(t) > x\} \leq f_2 \otimes g_2 (x-\alpha_2 \oslash \beta_2 (0)), \label{theorem1-2}
\end{equation}
\end{theorem}

\begin{proof}
We only prove the lower bound, since the upper bound can be easily proved following the same argument using the curves $\alpha_2$ and $\beta_2$, the bounding functions $f_2$ and $g_2$, and Inequality~(\ref{eq:lemma1-1}).

For any $t \geq s >0$, we have 
\begin{align}
E(t) =&   C(t) - C^*(t)  \nonumber \\
=& C(t) - C \otimes \beta_1(t) + C\otimes \beta_1(t) -C^*(t) \nonumber \\ 
  = &  \sup_{0\le s \le t} \{C(s,t) - \alpha_1(t-s) + \alpha_1(t-s) - \beta_1(t-s)\} \nonumber \\
  &+ C\otimes \beta_1(t) -C^*(t) \nonumber \\
 \geq &\sup_{0\le s \le t} \{\inf_{0\le s \le t}\{C(s,t) - \alpha_1(t-s)\}\nonumber \\ 
  & + \alpha_1(t-s) - \beta_1(t-s)\}
  + C\otimes \beta_1(t) -C^*(t) \nonumber \\
 = & \inf_{0\le s \le t}\{C(s,t) - \alpha_1(t-s)\} + \sup_{0\le s \le t} \{\alpha_1(t-s) \nonumber \\ 
 & - \beta_1(t-s)\} + C\otimes \beta_1(t) -C^*(t) \nonumber \\
  =& \inf_{0\le s \le t}\{C(s,t) - \alpha_1(t-s)\} + \sup_{t>0} \{\alpha_1(t) - \beta_1(t)\} \nonumber \\
  &+ C\otimes \beta_1(t) -C^*(t) \nonumber \\
   =& \inf_{0\le s \le t}\{C(s,t) - \alpha_1(t-s)\} + C\otimes \beta_1(t) -C^*(t)  \nonumber \\
  &+ \alpha_1\oslash \beta_1(0) \nonumber \\
\end{align}
If we consider $\inf_{0\le s \le t}\{C(s,t) - \alpha_1(t-s)\}$ and $C\otimes \beta_1(t) -C^*(t)$ as two random variables depending on $t$, Inequality~(\ref{theorem1-1}) follows according to the definitions of $s.e.c.$ and $s.e.d.$ curves and Inequality~(\ref{eq:lemma1-2}) in Lemma~\ref{lemma1}.
\end{proof}

\begin{remark}Theorem~\ref{theorem1} provides the following fundamental guidance in the energy management of systems powered with renewable energy: 
\begin{itemize}
\item Given the statistical feature of a renewable energy source and a safety threshold on remaining energy level, we can find a suitable energy consumption budget to service requirements (i.e., $\beta_1(t)$ in Theorem~\ref{theorem1}) based on (\ref{theorem1-1}), so that the system has sufficient energy with a certain probability.    
\item Given the statistical feature of a renewable energy source and the maximum capacity of energy storage in the system, we can obtain a minimal energy expenditure rate based on~(\ref{theorem1-2}), such that the probability that the system is overcharged is lower than a certain probability. 
\item Overall, given a renewable energy source, a safety threshold energy level, and an upper limit for energy storage, the two inequalities in Theorem~\ref{theorem1} describe the stochastically feasible region of energy budget for services.   
\end{itemize} 
\end{remark}

\section{Stochastic Service Guarantee of A Single Node} \label{sec:singleNode}
\subsection{System Model}

With the above energy models, we next study the problem of providing a stochastic service guarantee given an uncertain energy supply. We start with the single-node case. 

For the paper to be self-contained, we need to briefly introduce the core concepts in traditional stochastic network calculus-- stochastic traffic arrival curves and service curves~\cite{Jiangbook}. For traffic arrivals, we have the following model.
\nop{
\begin{definition}\label{def-tac}\textbf{The $t.a.c.$ model: }
A flow $A(t)$ is said to have a \textit{\underline{t}raffic-\underline{a}mount-\underline{c}entric (t.a.c.)} stochastic arrival curve $\alpha \in \mathcal{F}$ with bounding function $f\in \bar{\mathcal{F}}$, denoted by $$A\sim_{ta}<f,\alpha>,$$ if for all $t \ge s \ge 0$ and all $x\ge 0$, it holds~\cite{Jia,Jiangbook} 
\begin{equation}
Prob\{A(s,t)-\alpha(t-s) > x\} \le f(x).
\end{equation}
\end{definition}
}
\begin{definition}\label{def-vbc} \textbf{The $v.b.c.$ model: }
A flow $A(t)$ is said to have a \textit{\underline{v}irtual-\underline{b}acklog-\underline{c}entric (v.b.c.)} stochastic arrival curve $\alpha \in \mathcal{F}$ with bounding function $f\in \bar{\mathcal{F}}$, denoted by $$A\sim_{vb}<f,\alpha>,$$ if for all $t \ge 0$ and all $x\ge 0$, it holds~\cite{Jia,Jiangbook} 
\begin{equation}
Prob\{\sup_{0\le s \le t} [A(s,t)-\alpha(t-s)] > x\} \le f(x).
\end{equation}
\end{definition}

For service models, we have the followings.
\nop{
\begin{definition}\label{def-ws} \textbf{The $w.s.$ model: }
A server is said to provide a flow $A(t)$ with a \textit{\underline{w}eak \underline{s}tochastic (w.s.) service curve} $\beta \in \mathcal{F}$ with bounding function $g \in \bar{\mathcal{F}}$, denoted by $$S\sim_{ws}<g,\beta>,$$ if for all $t\ge 0$ and all $x\ge 0$, it holds~\cite{Jia,Jiangbook}  
\begin{equation}
Prob\{A\otimes\beta(t)-A^*(t) >x\} \le g(x). 
\end{equation}
\end{definition} 
}

\begin{definition}\label{def-sc} \textbf{The $s.c.$ model: }
A server is said to provide a flow $A(t)$ with a \textit{\underline{s}tochastic service \underline{c}urve (s.c.)} $\beta \in \mathcal{F}$ with bounding function $g \in \bar{\mathcal{F}}$, denoted by $$S\sim_{sc}<g,\beta>,$$ if for all $t\ge 0$ and all $x\ge 0$, it holds~\cite{Jia,Jiangbook}  
\begin{equation}
Prob\{\sup_{0\le s \le t}[A\otimes\beta(s)-A^*(s)] >x\} \le g(x). 
\end{equation}
\end{definition} 

\begin{remark}
The $s.c.$ model is adopted in this paper for ease of expressing the
results, particularly the concatenation property for the network case
analysis in Section~\ref{sec:network}. However, the $s.c.$ model may be too restrictive.
In the literature, a variation of the model is available, which is called
\textit{weak stochastic service curve (w.s)} in~\cite{Jia}. The $w.s.$ model is much less
restrictive and has been widely adopted in the stochastic network calculus
literature. We would like to stress that the $w.s.$ model can also be used
here, and it can be verified (e.g. see~\cite{Jia}) that the delay and backlog
bound results in this paper remain unchanged with the $w.s.$ model. The
affected one is the analysis on the network case, which would have much
complicated expression.
\end{remark}

\begin{definition}\label{def-ssc} \textbf{The $s.s.c.$ model: }
A server is said to provide a flow $A(t)$ with a \textit{\underline{s}trict \underline{s}tochastic service \underline{c}urve (s.s.c.)} $\beta \in \mathcal{F}$ with bounding function $g \in \bar{\mathcal{F}}$, denoted by $$S\sim_{ssc}<g,\beta>,$$ if during any period $(s, t]$ and for any $x\ge 0$, it holds~\cite{Jia,Jiangbook}  
\begin{equation}
Prob\{S(s,t)<\beta(t-s)-x\} \le g(x). 
\end{equation}
\end{definition} 

\begin{remark}
The $s.s.c.$ model is to decouple the service and traffic arrivals. It is useful in energy outage analysis as shown later in Section~\ref{sec:outage}.  
\end{remark}

\nop{
\begin{remark}Both stochastic traffic arrival curve and stochastic service curve are the extension of the corresponding deterministic case. For example, for task scheduling purpose, we can reserve the minimal service rate for a task and in this case, the service model becomes deterministic by setting the bounding function $g$ to $0$.   
\end{remark}
}

The following measures are of interest in service guarantee analysis under network calculus:
\begin{itemize}
\item The backlog $B(t)$ in the system at time $t$ is defined as:
\begin{equation}
B(t) = A(t) - A^*(t).
\end{equation}
\item The delay $D(t)$ at time $t$ is defined as:
\begin{equation} 
D(t)= \inf \{\tau \ge 0: A(t) \le A^*(t+\tau)\}. \label{delay}
\end{equation}
\end{itemize}

\subsection{Stochastic Service Guarantee}

To ease description, we define the following terms:
\begin{definition} \textbf{Power-rate function}~\cite{Zafer09} is a function that translates the amount of service to the amount of consumed energy. We use a generic notation $\mathcal{P}$ to denote the power-rate function.  
\end{definition}

\begin{definition}
\textbf{Energy-oblivious service curve} of a system is the service curve that the system would have if there were no energy constraint. 
\end{definition}

Note that most modern devices, such as computer servers and wireless transceivers, have the capability of adjusting processing/transmission rate. Associated with a rate, there is a corresponding power expenditure that is governed by the power rate function. Generally speaking, a low processing/transmission rate requires low energy consumption. The power-rate function is system dependent. For example, for most encoding schemes in wireless communication, the required power is a convex function of the rate~\cite{Zafer09}. In addition, there are some devices, e.g., Atmel microprocessors, whose energy consumption is dominated by its \textit{on} state. That is, as long as the device is powered on, it consumes energy at a roughly constant speed, no matter whether or not its CPU remains idle or executes tasks. To avoid those different details, we use a generic notation $\mathcal{P}$ to build a broadly-applicable model.  

Given a cumulative energy amount $C(t)$, the service amount that $C(t)$ can support, denoted as $S_C(t)$, can thus be calculated as:
\begin{equation} 
S_C (t) = \int_0^t \mathcal{P}^{-1}(C'(x)) dx,  \label{eq:SC}
\end{equation} 
Similarly, given a service amount $S(t)$, the corresponding energy consumption amount, denoted by  $C_S^*$, can be calculated as: 
\begin{equation} 
C_S^*(t) = \int_0^t \mathcal{P}(S'(x)) dx.  \label{eq:discharging}
\end{equation} 

\begin{remark} We assume that the inverse of power rate function $\mathcal{P}$, the derivative of the energy charging curve $C$, and the derivative of the service curve $S$, all exist. This assumption has been used and justified in~\cite{Zafer09}. This assumption is reasonable, because (1) there is usually a one-to-one mapping between the served data amount and the used energy amount, and (2) cumulative energy charging amount and cumulative service amount are mostly continuous. Nevertheless, even if the above assumption may not be true in some specific situations, numerical approximation can be used to estimate $S_C$ and $C_S^*$.  In the rest of the paper, we will use Equations~(\ref{eq:SC}) and~(\ref{eq:discharging}) without giving further explanation.
\end{remark}

With all notations being introduced, we are ready to answer the following question regarding a QoS guarantee:

Assume that a data flow $A(t)$ with the traffic arrival curve $A\sim_{vb}<f,\alpha>$ is input into a node, which has an energy-oblivious service curve $S\sim_{sc}<g,\beta>$. Assume that the system is powered by an energy source following $C(t) \sim_{sec}<f_1,\alpha_1, f_2, \alpha_2>$. For any time $t$, \textit{what are the stochastic bounds on the flow's delay $D(t)$ and backlog $B(t)$?}  

The main difficulty in solving the above question is that the data flow may not be able to get a service following the service curve $\beta$ due to the energy constraint. In other words, a service is possible only if the system has enough energy. By ``enough energy'', we mean that the system's energy level is above the safety threshold. To simplify presentation, we assume the safety threshold value is $0$ in this section. Otherwise, trivial modification is required for the following performance results.

We have the following important technical lemmas. 

\begin{lemma} \label{lemma-2-0}
Assume that functions $f_1, f_2, f_3 \in \mathcal{F}$, and assume that $f_1(0)=f_2(0)=f_3(0)=0$. For any $t\ge 0$, it holds that
\begin{equation}
f_1 \otimes (f_2\wedge f_3) (t) \le (f_1\otimes f_2 (t)) \wedge f_3(t)  \label{eq:lemma-2-0}
\end{equation}
\end{lemma}

\begin{lemma} \label{lemma-2}
If $X_1\geq X_2\geq 0$ and $X_3\geq X_4 \geq 0$, it holds 
\begin{equation}
X_1 \wedge X_4 - X_2 \wedge X_3 \leq X_1-X_2 + X_3-X_4. \label{eq:lemma-2}
\end{equation}
\end{lemma}

\begin{lemma}\label{lemma-2-3} \cite{Jiangbook}
Consider a random variable $X$. For any $x \ge 0$, $Prob\{[X]^+ > x\} = Prob\{X >x\}.$
\end{lemma}

\begin{theorem} \label{theorem2}
Assume that a node has an energy-oblivious service curve $S\sim_{sc}<g,\beta>$. Assume that the node is powered with an energy source following $C(t)\sim_{sec}<f_1, \alpha_1, f_2, \alpha_2>$.  The actual service available to the task, denoted by $S_e(t)$, follows $S_e(t) \sim_{sc} <\dot{g}, \beta \wedge \dot{\alpha_2}>$, where 
\begin{equation} 
\dot{g}(x) = g \otimes f_2(x),  \label{eq:g'}
\end{equation} 
\begin{equation}
\dot{\alpha_2}(t) = \int_{0}^t \mathcal{P}^{-1}(\alpha_2(x)) dx. \label{eq:dotAlpha}
\end{equation}
\end{theorem}
\begin{proof} It is clear that at any time $t\ge 0$, the actual service provided by the system, 
\begin{equation}
 S_e(t) = \left\{ 
  \begin{array}{l l}
    S(t) & \quad \text{if $S(t) \leq S_C(t)$ }\\
    S_C(t) & \quad \text{otherwise,}\\
  \end{array} \right.
\end{equation}
where $S_C$ is defined by~(ref{eq:SC}). This is because if $S(t) \le S_C (t)$, the energy constraint does not play a role, and the actual service amount is equal to $S(t)$. Otherwise, the actual service is equal to the amount allowed by the energy constraint, i.e., $S_C(t)$. In other words, $S_e(t) = S(t)\wedge S_C(t)$.

Next, we prove that $S_e(t) \sim_{sc} <\dot{g}, \beta \wedge \dot{\alpha_2}>$. Denote the arrival flow as $A(t)$ and the output flow as $A^*(t)$, which is equal to $S_e(t)$. For any time $0 \le s\le t$, based on Lemma~\ref{lemma-2-0} and Lemma~\ref{lemma-2}, we have 
\begin{align}
A\otimes (\beta \wedge \dot{\alpha_2}) (s) - &A^*(s) \le (A \otimes \beta (s)) \wedge \dot{\alpha_2} (s) - A^*(s) \nonumber \\
  &= (A \otimes \beta (s)) \wedge \dot{\alpha_2} (s) - S(s)\wedge S_C(s) \nonumber \\ 
  & \le [A\otimes \beta(s) - A^*(s)]^+ + [S_C(s) - \dot{\alpha_2}(s)]^+ \nonumber \\
\end{align}

Note that in the above, the first inequality is due to Lemma~\ref{lemma-2-0}, the second equality is because that $A^*(s)= S_e(s)$, and the last inequality is based on Lemma~\ref{lemma-2}.

We thus have 
\begin{align}
&\sup_{0\le s \le t} \{A  \otimes (\beta \wedge \dot{\alpha_2}) (s) - A^*(s)\} \nonumber \\
& \le \sup_{0 \le s \le t} \{[A\otimes \beta(s) - A^*(s)]^+ + [S_C(s) - \dot{\alpha_2}(s)]^+\} \nonumber \\
& \le \sup_{0 \le s \le t} \{[A\otimes \beta(s) - A^*(s)]^+\} + \sup_{0 \le s \le t}\{[S_C(s) - \dot{\alpha_2}(s)]^+\} \nonumber \\ 
\end{align}

Based on the definitions of the \textit{s.c.} curve and the \textit{s.e.c.} curve, and Lemma~\ref{lemma-2-3}, we have for any $x >0$, $0\le s\le t$,
\begin{equation} 
Prob\{\sup_{0\le s\le t} [A\otimes \beta(s) - A^*(s)]^+ > x\} \le g(x),\label{eq:theorem2-2}
\end{equation}
\begin{equation}
Prob\{\sup_{0\le s\le t} [S_C(s) - \dot{\alpha_2}(s)]^+ > x\} \le f_2(x).\label{eq:theorem2-3}
\end{equation}

For~(\ref{eq:theorem2-3}), we remind the readers that the calculations of $S_C$ and $\dot{\alpha_2}$ \emph{do not} change the bounding function. The theorem follows based on (\ref{eq:theorem2-2}), (\ref{eq:theorem2-3}), Inequality (\ref{eq:lemma1-1}), and the definition of \textit{s.c.} service curve. 
\end{proof}

We have the following theorem for service guarantee. 

\begin{theorem} \label{theorem3}
Assume that a data flow $A(t)$ with the traffic arrival curve $A\sim_{vb}<f,\alpha>$ is input into a node. Assume that the node has an energy-oblivious service curve $S\sim_{sc}<g,\beta>$. Assume that the node is powered with energy supply following $C(t)\sim_{sec}<f_1, \alpha_1, f_2, \alpha_2>$.
\begin{itemize}
\item  For any time $t\ge0, x\ge0$, 
\begin{equation} 
Prob\{D(t) > h(\alpha(t)+x, \beta \wedge \dot{\alpha_2}(t))\} \le f\otimes \dot{g}(x) 
\end{equation} 
where $\dot{g}$ is defined with (\ref{eq:g'}), $\dot{\alpha_2}$ is defined with (\ref{eq:dotAlpha}), and $h(\alpha(t)+x, \beta \wedge \dot{\alpha_2}(t))$ is the maximum horizontal distance between functions $\alpha(t)+x$ and $\beta \wedge \dot{\alpha_2}(t)$.
\item For any $t>0$ and $x\geq 0$, the backlog $B(t)$ is bounded by 
\begin{equation} 
Prob\{B(t) > x \} \le f\otimes \dot{g}(x-\alpha \oslash (\beta \wedge \dot{\alpha_2})(0)). 
\end{equation} 
\end{itemize}

\end{theorem}
   
\begin{proof} Theorem~\ref{theorem3} follows directly by applying Theorem~\ref{theorem2} to existing results of service guarantee in traditional stochastic network calculus (refer to Chapter $5$ of~\cite{Jiangbook}). 
\end{proof}   
   
\subsection{The Danger of Energy-Oblivious Service} \label{sec:outage} 

An argument for using deterministic energy management is that the system would be safe if the average energy depletion rate equals the predicted energy replenishment rate. Since no absolute guarantee can be made on the accuracy of energy prediction, there is a chance that an energy outage occurs, i.e., the remaining energy-level is lower than the safety threshold. In this section, we show that the energy outage probability is non-negligible. 

\begin{theorem} \label{theoremEnergaOutage}
Assume that a node provides an energy-oblivious service curve $S\sim_{ssc}<g,\beta>$. Assume that the node is powered with energy supply following $C(t)\sim_{sec}<f_1, \alpha_1, f_2, \alpha_2>$. Denote $C_S^*(t) = \int_0^t \mathcal{P}(S'(x)) dx$, and $E(t) = C(t)-C_S^*(t)$. For any time $t\ge0$, given the energy safety threshold $x \ge0$,
\begin{equation}
Prob\{E(t) < x\} \ge 1-f_2 \otimes g (x-\alpha_2 \oslash \dot{\beta} (0)), \label{eq:Energyoutage} 
\end{equation}
where $\dot{\beta}(t) = \int_0^t \mathcal{P}(\beta(x)) dx$.
\end{theorem}
\begin{proof}
For any $t \ge s \ge 0$, we have
\begin{align}
E(t) =& C(t) - C_S^*(t)  \nonumber \\
     = & C(t) -\alpha_2(t) + \dot{\beta}(t)-C_S^*(t) + \alpha_2(t)-\dot{\beta}(t)\nonumber \\
     \le & \sup_{0\le s \le t}\{C(s,t) -\alpha_2(t-s)\} \nonumber \\
     & + \sup_{0\le s \le t} \{\dot{\beta}(s,t)-C_S^*(s,t)\} 
      + \sup_{t\ge 0} \{\alpha_2(t)-\dot{\beta}(t)\}\nonumber \\
    =& \sup_{0\le s \le t}\{C(s,t) -\alpha_2(t-s)\}\nonumber \\
     & + \sup_{0\le s \le t} \{\dot{\beta}(s,t)-C_S^*(s,t)\} 
      + \alpha_2 \oslash \dot{\beta} (0). 
\end{align}
Based on Inequality~(\ref{eq:lemma1-1}), the definition of $s.e.c.$ curve, and the definition of $s.s.c.$ curve (we note again that the calculations of $\dot{\beta}$ and $C_S^*$ do not change the bounding function), we obtain for any $x>0$,
\begin{equation}
Prob\{E(t) > x\} \le f_2 \otimes g (x-\alpha_2 \oslash \dot{\beta} (0)).  
\end{equation}
The theorem follows.       
\end{proof}

As an example, assume that there exists an upper curve on the energy charging rate of $\alpha_2 (t) = \rho t + \sigma$, with a bounding function $f_2(x)= e^{-(x+2)}$, where $\sigma$ is a given safety threshold energy level and $\rho$ is the average energy charging rate. Assume that the node could provide energy-oblivious service $S~\sim_{s.c.} <\beta, g>$, where $g(x)= e^{-(x+2)}$. Assume that $\dot{\beta}(t)= \int_0^t \mathcal{P}(\beta(x)) dx =\rho t$. Clearly, the average energy charging rate is equal to the average energy discharging rate and initially the energy level at the node is safe. With Theorem~\ref{theoremEnergaOutage}, however, it is easy to calculate that the energy outage probability is no less than $72.9\%$, regardless of the time and the values of $\rho$ and $\sigma$.             

\begin{remark}
Theorem~\ref{theoremEnergaOutage} and the above example illustrate the danger of energy-oblivious service. Since a service is practically feasible only if it does not cause an energy outage, any service scheduling algorithm should consider the energy constraint. This is the main reason why in Section~\ref{sec:energyModel} the energy discharging model is inherently coupled with the energy charging model and why in our analysis for service guarantee, we enforce an energy constraint on the energy-oblivious service.   
\end{remark}

\section{Stochastic Service Guarantee for a Networked System} \label{sec:network}
In the previous section, we analyzed the single node case. Here, we study the service guarantee question in a network for which the network nodes are powered with a renewable energy source. It is well-known that we cannot simply add the delay bounds for each individual nodes to obtain the end-to-end delay bound along a path~\cite{CiucuScaling,Le}, since otherwise the network-wide delay bounds would be too loose to be useful. As a solution, the concatenation property of service curves should be proved and applied~\cite{Jiangbook}. 

Based on Theorem~\ref{theorem2} and the concatenation property of $s.c.$ service model, the following theorem holds immediately. 
\begin{theorem} \label{theorem4}
Consider a traffic flow passing through of a network of $N$ nodes in tandem. Assume that each node $i (=1, 2, \ldots, N)$ would provide an energy-oblivious service curve $S^i \sim_{sc} <g^i, \beta^i>$ to its input. Assume that each node $i$ is powered with an energy supply following $C^i \sim_{sec}<f^i_1, \alpha^i_1, f^i_2, \alpha^i_2>$. The network guarantees to the flow a stochastic service curve $S\sim_{sc} <g, \beta>$ with 
\begin{equation}
g = \dot{g}^{1}\otimes \dot{g}^{2} \ldots \otimes \dot{g}^{N}, \label{eq:theorem4_1} 
\end{equation}
\begin{equation}
\beta = \dot{\beta}^{1} \otimes \dot{\beta}^{2} \ldots \otimes \dot{\beta}^{N},\label{eq:theorem4_2} 
\end{equation}
where 
\begin{equation} 
\dot{g}^{i}(x) = g^i \otimes f^i_2(x), 
\end{equation} 
\begin{equation}
\dot{\beta}^{i} (t)= \beta^i(t) \wedge \int_{0}^t \mathcal{P}^{-1}(\alpha^i_2(x)) dx.
\end{equation}
\end{theorem}

\begin{remark} Theorem~\ref{theorem4} indicates that we can treat the concatenation of multiple nodes as a single system. To obtain the end-to-end performance, we just need to apply Theorem~\ref{theorem4} into Theorem~\ref{theorem3}, that is, replacing $\dot{g}$ and $\beta\wedge \dot{\alpha_2}(t)$ in Theorem~\ref{theorem3} with $g$ and $\beta$ in Theorem~\ref{theorem4}, respectively.
\end{remark}

\begin{remark} In Sections~\ref{sec:singleNode} and~\ref{sec:network}, we only used the upper curve of the charged energy amount to constrain a service. Actually, another set of similar analysis could be done, if we use the lower curve of the charged energy amount to constrain a service. The practical meaning of this constraint is to avoid overcharging the system and to effectively use the energy, i.e., the service should deplete the energy at a rate no less than the minimal charging rate. 
\end{remark}
 
\section{Nodes Powered with Multiple Energy Sources}  \label{sec:multipleSources}

It is possibly that a network node might be charged with multiple energy sources. For example, several types of energy such
as solar, wind, vibrational, and thermal among others can be scavenged from the surroundings of a sensor node to
replenish its battery~\cite{Roundy}. In this case, we need to study the statistical features of the integrated energy source, for which we have the following theorem. 

\begin{theorem} \label{theoremMultipleEnergy}
Assume that a node has $N$ energy sources, denoted by $C^1, C^2, \ldots, C^N$, with $C^i$ following $C^i(t)\sim_{sec}<f^i_1, \alpha^i_1, f^i_2, \alpha^i_2>$. Assume that the hardware permits the multiple energy sources to charge the battery simultaneously. The overall energy supply to the node, denoted by $C$, follows $C(t)\sim_{sec}<f_1,\alpha_1, f_2, \alpha_2>$, where 
\begin{equation}
f_1 (x) = f^1_1 \bar{\otimes} f^2_1 \ldots \bar{\otimes} f^N_1 (x)-(N-1), \label{eq:f_1} 
\end{equation}
\begin{equation}
\alpha_1 (t) = \alpha^1_1 + \alpha^2_1 \ldots + \alpha^N_1 (t), \label{eq:alpha_1} 
\end{equation}
\begin{equation}
f_2 (x) = f^1_2 \otimes f^2_2 \ldots \otimes f^N_2 (x), \label{eq:f_2}
\end{equation}
\begin{equation}
\alpha_2 (t) = \alpha^1_2 + \alpha^2_2 \ldots + \alpha^N_2 (t), \label{eq:alpha_2} 
\end{equation}
\end{theorem}
\begin{proof} We only prove Equations~(\ref{eq:f_1}) and~(\ref{eq:alpha_1}) by applying Inequality~(\ref{eq:lemma1-2}), as the proof for Equations~(\ref{eq:f_2}) and~(\ref{eq:alpha_2}) is similar by using Inequality~(\ref{eq:lemma1-1}). In addition, we only need to prove the case that $N=2$, because the theorem holds by recursively applying the result for $N=2$. 

As we assume that the multiple energy sources can charge the node simultaneously, for any time $t>0$, $C(t) = C^1(t) + C^2(t)$. For any $t \geq s >0$, we have 
\begin{align}
 & \inf_{0 \le s \leq t}\{C(s,t)-\alpha_1(t-s)\} \nonumber \\
 = & \inf_{0 \le s \leq t}\{C^1(s,t)+C^2(s,t)-(\alpha^1_1(t-s)+\alpha^2_1(t-s))\} \nonumber \\ 
  = &  \inf_{0 \le s \leq t}\{C^1(s,t)-\alpha^1_1(t-s)+(C^2(s,t)- \alpha^2_1(t-s))\} \nonumber \\
  \geq & \inf_{0 \le s \leq t}\{C^1(s,t)-\alpha^1_1(t-s)+\inf_{0 \le s \leq t}[C^2(s,t)- \alpha^2_1(t-s)]\} \nonumber \\
  =& \inf_{0 \le s \leq t}\{C^1(s,t)-\alpha^1_1(t-s)\} + \inf_{0 \le s \leq t}\{C^2(s,t)- \alpha^2_1(t-s)\} \label{eq:proofMultipleSource}
\end{align} 
The lower curve and the lower bounding function of $C(t)$, i.e., Equations~(\ref{eq:f_1}) and~(\ref{eq:alpha_1}), are thus proved based on the above inequality and Inequality~(\ref{eq:lemma1-2}).
\end{proof}
 
\section{Further Improvement on Bounds} \label{sec:improvement}
So far, we have developed an analytical framework based on stochastic network calculus to evaluate the performance of network systems with nodes powered by renewable energy sources. This analytical framework is general enough no matter whether or not the energy-oblivious service process and the energy charging process are independent, and no matter whether or not multiple energy sources, if they exists, are independent. If we assume that the energy-oblivious service process and the energy charging process are independent or that multiple energy sources are independent, which we believe is true in most applications, better performance bounds could be obtained, using the following lemma. 

\begin{lemma}\label{lemma-indep} 
Assume that the \textit{complementary cumulative distribution functions (CCDF)} of non-negative random variables $X, Y$ are $\bar{F}_{X} \in \bar{\mathcal{F}}$  and $\bar{F}_{Y} \in \bar{\mathcal{F}}$, respectively. Assume that $f_1(x) \leq \bar{F}_{X}(x) \leq f_2(x)$, $g_1(x) \leq \bar{F}_{Y}(x) \leq g_2(x)$, and $f_1, f_2, g_1, g_2 \in \bar{\mathcal{F}}$. Denote $Z=X + Y$. Assume that $X$ and $Y$ are independent, there holds for $\forall x \ge 0$,
\begin{equation} 
\bar{F}_Z(x) \leq 1- \bar{f}_2 \star \bar{g}_2 (x) , \label{eq:lemma-indep-1}
\end{equation} 
and 
\begin{equation} 
\bar{F}_Z(x) \geq 1- \bar{f}_1 \star \bar{g}_1 (x), \label{eq:lemma-indep-2}
\end{equation} 
where $\bar{f}_1= 1-[f_1]_1, \bar{f}_2= 1- [f_2]_1,\bar{g}_1= 1-[g_1]_1,  \bar{g}_2= 1-[g_2]_1$, and $\star$ is the Stieltjes convolution operation~\cite{Papoulis}. 
\end{lemma}

With Lemma~\ref{lemma-indep}, if the independence assumption could be made, the bounds in the theorems of this paper could be revised accordingly. That is, we could apply Lemma~\ref{lemma-indep} instead of Lemma~\ref{lemma1} in the proofs of the theorems to get better bounds. For example, if we assume that two energy sources are independent, the lower and upper bounds in Theorem~\ref{theoremMultipleEnergy} should be revised, respectively, to 
\begin{equation}
f_1 (x) = 1- (1-f^1_1) \star (1-f^2_1) , \label{eq:ind-f_1} 
\end{equation}
\begin{equation}
f_2 (x) = 1- (1-f^1_2) \star (1-f^2_2). \label{eq:ind-f_2}
\end{equation}
As an example, assume that $f^1_1(x) =f^2_1(x) = e^{-2x}$ and $f^1_2(x) =f^2_2(x) = e^{-x}$. With Theorem~\ref{theoremMultipleEnergy}, we have the lower bound $f_1(x) = e^{-2x} $ and the upper bound $f_2(x) = 2e^{-x/2}$. But, by applying  Lemma~\ref{lemma-indep}, we can get new lower bound  $f_1(x) =(1+2x)e^{-2x} $ and new upper bound $f_2(x) = (1+x) e^{-x}$. It is easy to verify that the new bounds are tighter as shown in Fig.~\ref{fig:tighterBound}. 

\begin{figure}
\begin{center}
  \includegraphics[width = 0.7\columnwidth]{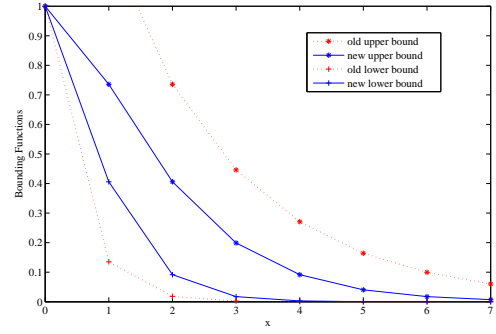}\\
  \caption{Comparison of old and new bounds}\label{fig:tighterBound}
  \vspace{-0.25in}
\end{center}
\end{figure}

\section{Related Work}
\label{sec:relatedwork}

With the increasing importance of greening computing, QoS and performance evaluation of network systems with uncertain energy supply have attracted much attention. Many performance modeling approaches have been developed, which could be divided into two main categories, deterministic and stochastic approaches.  

In the deterministic group, Zafer~\cite{Zafer09} \emph{et al.} use deterministic network calculus to model traffic arrival and traffic departure. They use a power-rate function to link the traffic departure rate and energy consumption rate. By considering the special features of specific power-rate functions, they formulate and solve the optimal transmission scheduling problem under the given energy constraints. Their work only focuses on single-node analysis and assumes that traffic arrivals and service rates are deterministic. Kansal \emph{et al.}~\cite{kansal2007power} propose a so-termed \textit{harvesting theory} to help energy management of a sensor node and determine the performance levels that the sensor node can support. The basic idea of the harvesting theory is to use a leaky bucket model to represent energy supply and energy depletion. Moser \emph{et al.}~\cite{moser2007real} describe energy-aware scheduling and prove the conditions for a scheduling algorithm to be optimal in a system whose energy storage is replenished \textit{predictably}. 

In the stochastic group, Markov chain models have been used extensively. Susu \emph{et al.}~\cite{susu2008stochastic} use a discrete-time Markov chain in which states represent different energy levels. Some work~\cite{niyato2007sleep} uses a Markov chain model to capture the influence of clouds and wind on solar radiation intensity. Relevant to stochastic energy modeling, there are many efforts to predict a stochastic energy supply. Lu~\emph{et al.}~\cite{lu2010accurate} assess three prediction techniques: regression analysis, moving average, and exponential smoothing. Recas~\emph{et al.}~\cite{recas2000hollows} propose a weather-conditioned moving average (\textsc{wcma}) model, which adapts to long-term seasonal changes and short-term sudden weather changes. Moser \emph{et al.}~\cite{moser2007real} introduce energy variability curves to predict the power provided by a harvesting unit. 

We develop our analytical framework based on stochastic network calculus~\cite{Ciucu06,Fidler,Jia,Jiangbook}. Unlike deterministic network calculus~\cite{Chang00,Le}, which searches for the worst-case performance bounds, stochastic network calculus tries to derive \textit{tighter} performance bounds, but with a small probability that the bounds may not hold true. Since most renewable energy sources, such as solar and wind energy, are not deterministic, stochastic network calculus is a good fit for the performance evaluation of systems using renewable energy. Nevertheless, traditional stochastic network calculus was not originally targeted at modeling such systems. Substantial work is thus required to extend this useful theory.

Recent interesting work by Wang \emph{et al.}~\cite{Wang11} uses stochastic network calculus to evaluate the reliability of the power grid with respect to renewable energy. Their energy supply and demand models are a subset of the models we present in Section~\ref{sec:energyModel}. Their work shows a good example of how to tailor our models for a specific application. Another \textit{fundamental} difference is that Wang \emph{et al.} define the energy supply and energy demand as two de-coupled random processes, while in our work energy discharging is inherently coupled with energy charging.       

Finally, related to analytical frameworks for performance modeling, there is a large body of research on energy-aware scheduling algorithms. For example, Niyato \emph{et al.}~\cite{niyato2007sleep} investigate the impact of different sleep and wake-up strategies on data communication among solar-powered wireless nodes. In~\cite{vigorito2007adaptive}, Vigorito \emph{et al.} propose an adaptive duty-cycling algorithm that ensures operational power levels at wireless sensor nodes regardless of changing environmental conditions. In~\cite{Gorlatova11}, Gorlatova \emph{et al.} measure the energy availability in indoor environment and based on the measurement results they  develop algorithms to determine energy allocation in systems with predictable energy inputs and in systems where energy inputs are stochastic. In the stochastic model, they assume that energy inputs are i.i.d. random variables. Unlike the above work, our analytical framework is generic and uses only abstract notations on energy charging/discharging amount, traffic arrival amount, etc. As such, all the above work could be treated as a special case of our more general framework in which the abstract functions are replaced with concrete ones for specific applications. 

\section{Conclusion and Future Work} \label{sec:conclusion}

The application of renewable energy poses many challenging questions when it comes to a QoS guarantee for IT infrastructure. Particularly, the strong call for a generic analytical tool for the performance modeling and evaluation of such systems remains unanswered. We, in this paper, develop such a theory using a stochastic network calculus~\cite{Ciucu06,Fidler,Jia,Jiangbook} approach. We enrich the existing stochastic network calculus theory to make it useful for evaluating systems with a stochastic energy supply. We introduce new models to capture the dynamics in the energy charging and discharging processes, and derive new analytical results to provide fundamental performance bounds, which could be used to guide the energy management and the design of task scheduling algorithms. 

This paper mainly focuses on the introduction of a theoretical framework. Along the line, we envisage many future research topics, including (1) the application of energy charging/discharging models in different applications (e.g.,~\cite{susu2008stochastic,Wang11}), (2) analyzing the service model of different scheduling strategies (e.g.,~\cite{moser2007real}), and (3) network reliability analysis (e.g., network-wide energy outage analysis).    

\section*{Acknowledgment}
Kui Wu and Dimitri Marinakis were supported by the Natural Sciences and Engineering Research Council of Canada (NSERC). Yuming Jiang was supported by the Centre for Quantifiable Quality of Service (Q2S) in Communication Systems, Centre of Excellence, which is appointed by the Research Council of Norway, and funded by the Research Council, NTNU and UNINETT. 

\bibliographystyle{abbrv}
\bibliography{reference}
\section*{Appendix}
\balance 
\subsection{Proof of Lemma~\ref{lemma1}}
\begin{proof} The proof of (\ref{eq:lemma1-1}) could be found at~\cite{Jia}. We only prove (\ref{eq:lemma1-2}). For any $y \geq x \geq 0$, 
\begin{equation} \{X + Y \le y\} \cap \{X > x\} \cap \{Y > y-x\} = \emptyset,\end{equation} 
where $\emptyset$ denotes the null set. We thus have 
\begin{equation} \{X+ Y > y\} \supset (\{X > x\} \cap \{Y > y-x\})\end{equation} 
and hence 
\begin{align}
& Prob\{X+Y > y\} \geq Prob\{\{X > x\} \cap \{Y > y-x\}\} \nonumber \\
\geq & 1- Prob\{X \le x\} - Prob\{Y \le y-x\}\nonumber \\
 = & Prob\{X > x\} + Prob\{Y > y-x\} -1 . \nonumber \\
\end{align}
Since the above inequality holds for all $y \geq x \geq 0$, we get 
\begin{align}
Prob\{X+Y > y\} &\geq \sup_{y \geq x \geq 0} \{Prob\{X > x\} \nonumber  \\
&+ Prob\{Y > y-x\}\} -1,
\end{align}
which is 
\begin{align}
\bar{F}_Z(y) \geq \bar{F}_{X} \bar{\otimes} \bar{F}_{Y}(y) -1 \geq  f_1 \bar{\otimes} g_1(y) -1. 
\end{align}
\end{proof}

\subsection{Proof of Lemma~\ref{lemma-2-0}}
\begin{proof}
By definitions of the $\otimes$ and $\wedge$ operations, we need to prove that 
\begin{equation}
\inf_{0\le s \le t} \{f_1(s) + f_2\wedge f_3 (t-s)\} \le \inf_{0\le s \le t} \{f_1(s) + f_2(t-s)\} \wedge f_3(t). \label{eq:lemma-2-1}
\end{equation}

Since $f_3 \in \mathcal{F}$, we have 
\begin{equation}
\inf_{0\le s \le t} \{f_1(s) + f_2\wedge f_3 (t-s)\} \le \inf_{0\le s \le t} \{f_1(s) + f_2 (t-s) \wedge f_3(t)\}.
\end{equation}
We therefore only need to prove that 
\begin{equation}
\inf_{0\le s \le t} \{f_1(s) + f_2 (t-s) \wedge f_3(t)\} \le \inf_{0\le s \le t} \{f_1(s) + f_2(t-s)\} \wedge f_3(t). \label{eq:lemma-2-2}
\end{equation}
First, for any given $t\ge 0$, there exists a $s_1 (0 \le s_1 \le t)$, such that 
\begin{equation}
\inf_{0\le s \le t} \{f_1(s) + f_2(t-s)\} = f_1(s_1) + f_2(t-s_1). \label{eq:lemma2-overall} 
\end{equation}
In addition, when $t$ is given, $f_3(t)$ is a constant. 

We next prove that Inequality~(\ref{eq:lemma-2-2}) holds true for the following two exclusive cases:
\begin{itemize}
\item Case 1: $f_1(s_1) + f_2(t-s_1) \le f_3(t).$ The right hand side of Inequality~(\ref{eq:lemma-2-2}) equals $f_1(s_1) + f_2(t-s_1)$. Since $f_1 \in \mathcal{F}$ and $f_1(0) =0$, we have $f_1(s_1) \ge 0$. Therefore,
\begin{equation}
f_2(t-s_1) \le f_3(t). \label{eq:lemma2-case1}
\end{equation}
We need to show that the left hand side of Inequality~(\ref{eq:lemma-2-2}) is no larger than $f_1(s_1) + f_2(t-s_1)$. For the given $t\ge 0$, there exists a $s_2 (0 \le s_2 \le t)$ such that  
\begin{equation}
\inf_{0\le s \le t} \{f_1(s) + f_2 (t-s) \wedge f_3 (t)\} = f_1(s_2) + f_2 (t-s_2) \wedge f_3 (t).\label{eq:lemma2-case1-2}
\end{equation} 
We have  
\begin{align}
f_1(s_2) + f_2 (t-s_2) \wedge f_3 (t) & \le f_1(s_1) + f_2 (t-s_1) \wedge f_3 (t) \nonumber \\
 = & f_1(s_1) + f_2 (t-s_1). \nonumber \\
\end{align}
In the above, the first inequality is due to~(\ref{eq:lemma2-case1-2}), and the second equality is due to~(\ref{eq:lemma2-case1}).
\item Case 2: $f_1(s_1) + f_2(t-s_1) > f_3(t).$ The right hand side of Inequality~(\ref{eq:lemma-2-2}) equals $f_3(t)$. We also have 
$f_1(0) + f_2 (t-0) \ge f_1(s_1) + f_2 (t-s_1)$ due to~(\ref{eq:lemma2-overall}). Therefore, 
$$f_2(t) \ge f_1(s_1) + f_2 (t-s_1) > f_3(t),$$ and 
\begin{align}
\inf_{0\le s \le t}\{f_1(s) + f_2 (t-s) \wedge f_3 (t) \} & \le f_1(0) + f_2 (t) \wedge f_3 (t) \nonumber \\
 &=  f_3(t)\nonumber \\
\end{align}
\end{itemize}
In conclusion, Inequality~(\ref{eq:lemma-2-2}) is true for both cases. The lemma is proved. 
\end{proof}

\subsection{Proof of Lemma~\ref{lemma-2}}
\begin{proof} 
To evaluate the value of $X_1 \wedge X_4 - X_2 \wedge X_3$, we have: 
\begin{itemize}
\item Case 1: $X_1 \leq X_4$ and $X_2\geq X_3$. We have $X_1=X_2=X_3=X_4$, because $X_1\geq X_2\geq 0$ and $X_3\geq X_4 \geq 0$. Thus (\ref{eq:lemma-2}) holds.
\item Case 2: $X_1 \leq X_4$ and $X_2 \leq X_3$. We have  $$X_1 \wedge X_4 - X_2 \wedge X_3 = X_1 - X_2 \leq X_1-X_2 + X_3-X_4.$$
\item Case 3: $X_1 \geq X_4$ and $X_2\geq X_3$. We have 
$$X_1 \wedge X_4 - X_2 \wedge X_3 = X_4 - X_3 \leq 0.$$ 
Inequality (\ref{eq:lemma-2}) holds since the right side is no less than $0$. 
\item Case 4: $X_1 \geq X_4$ and $X_2\leq X_3$. We have
$$X_1 \wedge X_4 - X_2 \wedge X_3 \leq X_1 - X_2 \leq X_1-X_2 + X_3-X_4. $$ 
\end{itemize}
The lemma is proved since the above list covers all possible scenarios. 
\end{proof}

\subsection{Proof of Lemma~\ref{lemma-indep}}

\begin{proof} The proof of (\ref{eq:lemma-indep-1}) could be found at~\cite{Papoulis}. We only prove (\ref{eq:lemma-indep-2}). 

For independent non-negative random variables, $X$ and $Y$, we have $\forall x \ge 0$, 
$$F_Z(x) = \int_0^x F_{X}(x-y) dF_Y(y).$$
Note that $F_X, F_Y, \bar{f}_1, \bar{g}_1$ are wide-sense increasing, $F_X \leq \bar{f}_1$ and $F_Y \leq \bar{g}_1$. Hence, we have 
\begin{align}
F_Z(x) & = \int_0^x F_{X}(x-y) dF_Y(y) \leq \int_0^x \bar{f}_1(x-y) dF_Y(y)\nonumber \\
 = & \int_0^x F_Y(x-y) d\bar{f}_1(y)  \nonumber \\
 \leq & \int_0^x \bar{g}_1(x-y) d\bar{f}_1(y) = \bar{f}_1 \star \bar{g}_1 (x).
\end{align}
Note that the third equality holds because Stieltjes convolution operation is commutative. Inequality (\ref{eq:lemma-indep-2}) is thus proved since $\bar{F}_Z(x) = 1- F_Z(x)$. 
\end{proof}

\end{document}